\documentclass[preprint,12pt]{elsarticle}

\usepackage{epsfig,epic,eepic,units}
\usepackage{hyperref}
\usepackage{url}
\usepackage{longtable}
\usepackage{mathrsfs}
\usepackage{multirow}
\usepackage{bigstrut}
\usepackage{amssymb}
\usepackage{graphicx}
\usepackage{capt-of}
\usepackage{enumerate}	

\usepackage{amsmath}
\usepackage{amsfonts}
\usepackage[T1]{fontenc}

\usepackage{amsthm}

\newtheorem{theorem}{Theorem}

\newtheorem{lemma}{Lemma}

\newtheorem{relation}{Relation}

\newtheorem{definition}{Definition}

\usepackage{algorithm}
\usepackage{algorithmic}

\usepackage{enumitem}
\usepackage{booktabs}
\usepackage{soul}

\newcommand{\Avoid}{\mathcal{O}}
\newcommand{\avoid}{o}

\usepackage{multirow}
\usepackage{float}
\usepackage[caption = false]{subfig}

\journal{Linear Algebra and its Applications}

\begin{document}

\begin{frontmatter}



\title{Avoidance Markov Metrics and Node Pivotality Ranking}


\author{Golshan Golnari, Zhi-Li Zhang, and Daniel Boley}

\address{Computer Science and Engineering Department, University of Minnesota}

\begin{abstract}
We introduce the \textit{avoidance} Markov metrics and theories which provide more flexibility in the design of random walk and impose new conditions on the walk to \textit{avoid} (or \textit{transit}) a specific node (or a set of nodes) before the stopping criteria. These theories help with applications that cannot be modeled by classical Markov chains and require more flexibility and intricacy in their modeling. Specifically, we use them for the \textit{pivotality} ranking of the nodes in a network reachabilities. 
More often than not, it is not sufficient
simply to know whether a source node $s$ can reach a target node $t$
in the network and additional information associated with reachability,
such as how long or how many possible ways node $s$ may take to
reach node $t$, is required. In this paper, we analyze the pivotality of the nodes which capture how pivotal a role that a node $k$ or a subset
of nodes $S$ may play in the reachability from node $s$ to node $t$ in a
given network.
Through some synthetic and
real-world network examples, we show that these metrics build a
powerful ranking tool for the nodes based on their pivotality in the reachability.

\end{abstract}


\end{frontmatter}



\section{Introduction}
Markov metrics have proved to be a powerful tool for analyzing and solving a variety of problems \cite{qian2005power,steele2001random,ramage2009random,wang2012deterministic}. Hitting time, for instance, as the most well-known Markov metric has been exploited vastly in different network analysis applications. Fouss et al. \cite{fouss2007random} used the hitting time (and commute time) as the measure of similarity between nodes in a recommendation system. Sarkar et al. \cite{sarkar2008fast} developed a fast proximity search in large networks by means of hitting time. Chen et al. \cite{chen2008clustering} presented a clustering algorithm via hitting time on directed graphs. However, all of these works are confined to the applications of the \textit{classical} hitting time. 
 The existing theory on classical Markov metrics, including fundamental matrix, hitting time, hitting cost, and absorption (hitting) probability, is the result of imposing only the stopping criteria on the Markov chain (or equivalently on the random walk), which is being absorbed by the absorbing state  (or hitting the target node for the first time), and has no control or condition on the visiting states in the middle of the transition (walk).

In this paper, we introduce \textit{avoidance} Markov metrics which provide more flexibility in the design of random walk and impose new conditions on the walk to \textit{avoid} (or \textit{transit}) a specific node (or a set of nodes) before the stopping criteria. In particular, we introduce avoidance fundamental matrix, avoidance hitting time, transit hitting time, and avoidance hitting cost and 
establish that they can be computed from the  fundamental matrices associated with the appropriately defined random walk transition
probability matrices. We also provide useful relations and lemmas for computing the avoidance metrics. Thereafter, we show that these new metrics can provide a powerful tool for pivotality ranking of the nodes in network analysis applications. 

\textbf{Pivotality Ranking:} Each rechability in a network has some additional information associated with it, such as how long or how many possible ways are connecting the source node to target node, which are essential for application such as packet routing, flow scheduling, load balancing, and power management. We propose a pivotality metric to capture how pivotal a role that a third node (or a subset of nodes)
may play in the reachability from a source node to a target node in a given
network by quantifying how many (and how long) paths from source to target go through that third node, and how many do not. We present a few network examples for which other pivotality-type of metrics, in contrast to ours, fail to rank the nodes according to their pivotal roles for a chosen reachability. We also present two real-world networks to show how well our metric can capture the correct and intuitive pivotality ranking of the nodes.

The rest of this paper is organized as follows. A preliminary on classical Markov metrics is presented in Section (\ref{sec:prelim}). We introduce the avoidance Markov metrics, including avoidance fundamental matrix, avoidance hitting time, transit hitting time, and avoidance hitting cost in Section (\ref{sec:avoid}). Next, we provide useful relations, lemmas, and theorems on avoidance Markov metrics which would be insightful for future studies as well.
At the end, we present an application of avoidance Markov metrics in pivotality ranking of nodes in reachability problems in Section (\ref{sec:pivotal}). 
 
\section{Classical Markov Metrics} \label{sec:prelim}

Let $G={(X_k)}_{k>0}$ be a discrete-time Markov chain with transition matrix $P$. A Markov chain is called absorbing if it has at least one absorbing state that, once entered, cannot be left. The other states of an absorbing chain, that are not traps, are called non-absorbing or transient states. In an absorbing Markov chain, from each transient state at least one absorbing state should be reachable. Assuming that states are ordered in the way that set of transient
states $\mathcal{T}$ come first and set of absorbing states $\mathcal{A}$ come last, the transition matrix for an absorbing Markov chain takes the following block matrix form:
\begin{equation}
\label{eq:P}
P=\left[
\begin{array}{ c c }
P_{\mathcal{T}\mathcal{T}} & P_{\mathcal{T}\mathcal{A}} \\
0 & I_{\mathcal{A}\mathcal{A}}
\end{array} \right],
\end{equation}
where $I_{\mathcal{A}\mathcal{A}}$ is an $|\mathcal{A}|\times|\mathcal{A}|$ identity matrix and $P$ is row-stochastic.

Let indicator function $1_{\{X_k=m\}}$ be a random variable equal
to 1 if $X_k=m$ and 0 otherwise. The number of visits $\nu_m$ to
$m$ is written in terms of indicator functions as $\nu_m=\sum_{k=0}^{\infty} 1_{\{X_k=m\}}$. The $(s,m)$ entry of fundamental matrix $F^{\mathcal{A}}$ represents the expected number of visits to $m$ when the chain starts at $s$ and before absorption by ${\mathcal{A}}$ \cite{norris1998markov}:

\begin{equation}
F^{\mathcal{A}}_{sm}=\mathbb{E}_s(\nu_m)
\end{equation}

In the matrix form, fundamental matrix is computed as follows \cite{golnari2017random}: $F^{\mathcal{A}}=(I-P_{\mathcal{T}\mathcal{T}})^{-1}$

The hitting time, also known as the absorption time, is a random variable $\kappa_{\mathcal{A}}:\Omega\rightarrow\{0,1,2,...\}\cup\{\infty\}$
given by
$\kappa_{\mathcal{A}}=\inf{\{\kappa\geq 0: X_{\kappa}\in\mathcal{A}\}}$,
where we agree that the infimum of the empty set $\emptyset$ is $\infty$ \cite{norris1998markov}. The hitting time $\kappa_{\mathcal{A}}$ represents the number of steps that the walk takes until it hits $\mathcal{A}$ for the first time, and its expected value when the walk starts at $s$ is denoted by \cite{norris1998markov}: 
\begin{equation} \label{eq:H_stoc}
H_s^{\mathcal{A}}=\mathbb{E}_s[\kappa_\mathcal{A}]
\end{equation}
(Expected) hitting time can be computed from the fundamental matrix \cite{golnari2017random}:
$H_s^{\mathcal{A}}=\sum_m F_{sm}^{\mathcal{A}}$
  
As a generalization to hitting time, hitting cost accounts for the cost of the transitions as well. The hitting cost is a random variable $\eta_{\mathcal{A}}:\Omega\rightarrow \mathcal{C}$ given by
$\eta_{\mathcal{A}}=\inf{\{\eta\geq 0: \exists k, X_k\in\mathcal{A}, \sum_{i=1}^k w_{X_{i-1}X_i}=\eta\}}$, where  $\mathcal{C}$ is a countable set, $w_{ij}$ is the cost of edge $e_{ij}$, and we agree that the infimum of the empty set $\emptyset$ is $\infty$ \cite{golnari2017random}. The hitting cost $\eta_{\mathcal{A}}$ represents the total cost of the transitions that the chain takes until it gets absorbed by $\mathcal{A}$, and its expected value when the chain starts at $s$ is denoted by: 
\begin{equation}
U_s^{\mathcal{A}}=\mathbb{E}_s[\eta_{\mathcal{A}}]
\end{equation}
(Expected) hitting cost was first introduced by Fouss et al. \cite{fouss2007random}, presented in a recursive form: $U_s^{\mathcal{A}}=r_s + \sum_{m\in \mathcal{N}_{out}(s)}P_{sm}U_m^{\mathcal{A}}$, where $r_s$ is the expected out-going cost $r_s=\sum_i p_{si}w_{si}$. Later, Golnari et al. \cite{golnari2017random} provided a closed-form formulation to compute the (expected) hitting cost from the fundamental matrix: $U_s^{\mathcal{A}}=\sum_m F_{sm}^{\mathcal{A}}r_m$.

The absorption probability represents the probability that the chain ends up with each of the absorbing states and is denoted by $Q$ which is a $|\mathcal{T}|\times|\mathcal{A}|$ matrix \cite{snell}:
\begin{equation} \label{eq:Q}
Q^{\mathcal{A}}=FP_{\mathcal{T}\mathcal{A}}
\end{equation}
The $(s,t)$-th entry of $Q$
is the probability of absorption by absorbing state $t$ when the chain starts from transient state $s$. 
  We denote this entry by $Q_s^{\{t,\overline{\Avoid}\}}$, where $\Avoid=\mathcal{A}\setminus\{t\}$, to be more clear about the absorbing state which is hit ($t$) and the ones that are not hit ($\Avoid$). Note that $\sum_{i\in\mathcal{A}} Q_s^{\{i,\overline{\mathcal{A}\setminus\{i\}}\}}=1$, since starting from any state $s$, the chain will end up being absorbed by one of the absorbing states eventually.
To learn more about the classical Markov metrics, please refer to \cite{golnari2017random}.

\subsection{Connection to Networks}
Consider a weighted and directed network denoted by $G=(V,E,A)$, where $V$ is the set of nodes, $E$ is the set of edges, and $A$ is the adjacency matrix whose $a_{ij}$ entry indicates the distance from $i$ to $j$ if edge $e_{ij} \in E$, otherwise $a_{ij}=0$. A random walk over $G$ is modeled by a Markov chain, where the nodes of $G$ represent the states of the Markov chain and the Markov chain is fully described by its transition probability matrix: $P = D^{-1}A$, where $D=diag[d_i]$ is the diagonal matrix of $d_i$'s, and $d_i=\sum_{j}a_{ij}$ is referred to the (out-)degree of node $i$. In addition, the target nodes in $G$ can be represented as absorbing states in the Markov chain as once being hit, the random walk stops walking around. Throughout the paper, the words ``node" and ``state", and ``random walk over a network" and ``Markov chain" are used interchangeably.

\section{Avoidance Markov Metrics} \label{sec:avoid}
 
In this section, we introduce four advanced Markov metrics with modified properties and conditions. 

\begin{definition}[Avoidance fundamental matrix]
	The avoidance fundamental matrix is the conditional expected number of visits to $m$ while avoiding $\avoid$, when the Markov chain starts from $s$ and before absorption by $t$:
	\begin{equation}
	F_{sm}^{\{t,\overline{\avoid}\}}=\mathbb{E}_s(\nu_m|X_{i\leq \kappa_t}\neq \avoid),
	\end{equation}
	where $\nu_m$ represents the number of visits to
	$m$. 
\end{definition}
	
\begin{theorem}
	The avoidance fundamental matrix can be computed from the classical Markov metrics:
	\begin{equation} \label{eq:avoidanceF}
	F_{s,m}^{\{t,\overline{\avoid}\}}=F_{s,m}^{\{t,\avoid\}}.\frac{Q_m^{\{t,\overline{\avoid}\}}}{Q_s^{\{t,\overline{\avoid}\}}}
	\end{equation}
\end{theorem}

\begin{definition}[Avoidance hitting time]
	Avoidance (expected) hitting time from $s$ to $t$ avoiding node $\avoid$ is the conditional expectation over the number of steps required to hit $t$ for the first time when starting from $s$ and conditioned on avoiding $\avoid$ on the way:
	\begin{equation}
	H_s^{\{t,\overline{\avoid}\}}=\mathbb{E}_s[\kappa_t|X_{i\leq \kappa_t}\neq \avoid],
	\end{equation}
	where $\kappa_{t}$ is a random variable which represents the number of transitions until being absorbed by $t$.
\end{definition}
\begin{theorem}
	The avoidance hitting time can be computed from the classical Markov metrics:
	\begin{equation} \label{eq:avoidanceH}
	H_{s}^{\{t,\overline{\avoid}\}}=\sum_m F_{s,m}^{\{t,\avoid\}}.\frac{Q_m^{\{t,\overline{\avoid}\}}}{Q_s^{\{t,\overline{\avoid}\}}}
	\end{equation}
\end{theorem}

To find the more general form of (\ref{eq:avoidanceF}) and (\ref{eq:avoidanceH}) for a \textit{set} of avoiding nodes and the corresponding proofs, please refer to (\ref{apx:avoidFundMat}) and (\ref{apx:avoidHitTime}) respectively.

\begin{definition}[Avoidance hitting cost]
	Avoidance (expected) hitting cost from $s$ to $t$ avoiding node $\avoid$ is the conditional expectation over the cost of steps required to hit $t$ for the first time when starting from $s$ and conditioned on avoiding $o$ on the way:
	\begin{equation}
	U_s^{\{t,\overline{\avoid}\}}=\mathbb{E}_s[\eta_t|X_k=t, X_{i\leq k}\neq \avoid]
	\end{equation}
	where $\eta_{t}$ is a random variable which represents the cost of transitions until being absorbed by $t$.

\end{definition}
The more general form of avoidance hitting cost is presented in (\ref{apx:avoidHitCost}). 
\begin{theorem}
	The avoidance hitting cost can be computed from classical Markov metrics:
	\begin{equation} \label{eq:avoidanceU}
	U_{s}^{\{t,\overline{\avoid}\}}=\sum_m (F_{s,m}^{\{t,\avoid\}}.\frac{ Q_m^{\{t,\overline{\avoid}\}}}{Q_s^{\{t,\overline{\avoid}\}}})r_m^{\{t,\overline{\avoid}\}},
	\end{equation}
	where $r_m^{\{t,\overline{\avoid}\}}=\sum_i p_{mi}w_{mi}\frac{Q_i^{\{t,\overline{\avoid}\}}}{Q_m^{\{t,\overline{\avoid}\}}} $.
\end{theorem}
Please find the proof for the more general form of avoidance hitting cost in (\ref{apx:avoidHitCost}).

Closely related to avoidance hitting time, we define the notion of transit hitting time: 
\begin{definition}[Transit hitting time]
	For any third node $\avoid$, the transit hitting time $H_{s}^{\{t,\breve{\avoid}\}}$
	is
	the expected number of transitions which starts
	at $s$ and always traverses node $\avoid$ before being absorbed by $t$, obtained from:
	\begin{equation} \label{eq:transitH}
	H_{s}^{\{t,\breve{\avoid}\}}=H_{s}^{\{\avoid,\overline{t}\}}+H_\avoid^{\{t\}}
	\end{equation}
\end{definition}

We also show that the classical hitting time can be decomposed into two terms of avoidance hitting time and transit hitting time,  with respect to a third node. For more details please refer to Theorem~\ref{thm:hittime_decompose}.

\subsection{Useful Relations} 
In the following, we present further relations and insights for advanced Markov metrics to shed light on their usefulness for Markov analysis.  

\begin{relation} \label{re:HF1}
	Similar to classical counterpart, avoidance hitting time is obtained by row-sum of avoidance fundamental matrix:
	\begin{equation}
	H_{s}^{\{t,\overline{\avoid}\}}=\sum_m F_{s,m}^{\{t,\overline{\avoid}\}}
	\end{equation}
\end{relation}
\begin{proof}
	Compare (\ref{eq:avoidanceF}) and (\ref{eq:avoidanceH}).
\end{proof}

\begin{relation}
	Similar to classical counterpart, avoidance hitting cost is obtained by weighted row-sum of avoidance fundamental matrix:
	\begin{equation}
	U_{s}^{\{t,\overline{\avoid}\}}=\sum_m F_{s,m}^{\{t,\overline{\avoid}\}}r_m^{\{t,\overline{\avoid}\}},
	\end{equation}
	where $r_m^{\{t,\overline{\avoid}\}}=\sum_i p_{mi}w_{mi}\frac{Q_i^{\{t,\overline{\avoid}\}}}{Q_m^{\{t,\overline{\avoid}\}}} $.
\end{relation}
\begin{proof}
	Compare (\ref{eq:avoidanceF}) and (\ref{eq:avoidanceU}).
\end{proof}

\begin{relation} \label{lem:H2target_decomposition}
	Decomposing the classical hitting time for two target nodes into individual target avoidance components yields:
	\begin{equation} H_s^{\{t,\avoid\}}=Q_s^{\{t,\overline{\avoid}\}}H_s^{\{t,\overline{\avoid}\}}+Q_s^{\{\avoid,\overline{t}\}}H_s^{\{\avoid,\overline{t}\}}
	\end{equation}
\end{relation}
\begin{proof}
	Use \ref{eq:avoidanceH} and the fact that $Q_s^{\{t,\overline{\avoid}\}}+Q_s^{\{\avoid,\overline{t}\}}=1$
\end{proof}

\begin{theorem}[Hitting time decomposition into transit and avoidance components] \label{thm:hittime_decompose}
	The  hitting time from node $s$ to node $t$ can be decomposed into an ``avoidance" hitting time component and a ``transit" hitting time component with respect to any node $\avoid$ as follows:
	\begin{equation} 
	H_s^{\{t\}}=Q_s^{\{t,\overline{\avoid}\}}H_s^{\{t,\overline{\avoid}\}}+Q_s^{\{\avoid,\overline{t}\}}H_s^{\{t,\check{\avoid}\}}. 
	\end{equation}
\end{theorem}
\begin{proof}
		Take sum over $m$ for both sides of Lemma (\ref{Nlemma1}) and use Lemma (\ref{Nlemma3}) to obtain the following equation:
	\begin{equation} \label{eq:Hjk_1}
	H_{i}^{\{j,k\}}=H_i^{\{j\}}-Q_{i}^{\{k,\overline{j}\}}H_k^{\{j\}}
	\end{equation}

	Substituting Relation (\ref{lem:H2target_decomposition}) in Eq. (\ref{eq:Hjk_1}) yields the following relation:
	\begin{equation}
	H_i^{\{j\}}=Q_i^{\{j,\overline{k}\}}H_i^{\{j,\overline{k}\}}+Q_i^{\{k,\overline{j}\}}(H_i^{\{k,\overline{j}\}}+H_k^{\{j\}})=Q_i^{\{j,\overline{k}\}}H_i^{\{j,\overline{k}\}}+Q_i^{\{k,\overline{j}\}}H_i^{\{j,\check{k}\}}, \nonumber
	\end{equation} 
	where $H_{i}^{\{j,\check{k}\}}=H_{i}^{\{k,\overline{j}\}}+H_k^{\{j\}}$. 
	
\end{proof}

\begin{relation} \label{re:Fo}
	Avoidance fundamental matrix $F_{sm}^{\{t,\overline{\avoid}\}}$ can be written in terms of classical fundamental matrix $F^{\{\avoid\}}$ where the avoiding node $\avoid$ is the only absorbing state: 
	\begin{equation} \label{eq:Favoid_Fo} F_{sm}^{\{t,\overline{\avoid}\}}=F_{mt}^{\{\avoid\}}(\frac{F_{sm}^{\{\avoid\}}}{F_{st}^{\{\avoid\}}}-\frac{F_{tm}^{\{\avoid\}}}{F_{tt}^{\{\avoid\}}})
	\end{equation}
\end{relation}
\begin{proof}
	Apply Lemmas (\ref{Nlemma1}) and (\ref{Nlemma3}) in Eq. (\ref{eq:avoidanceF}).
\end{proof}

\begin{relation} \label{re:Ft}
	Avoidance fundamental matrix $F_{sm}^{\{t,\overline{\avoid}\}}$ can be written in terms of classical fundamental matrix $F^{\{t\}}$: 
	\begin{eqnarray} F_{sm}^{\{t,\overline{\avoid}\}}&=&\frac{1}{F_{\avoid\avoid}^{\{t\}}}\frac{(F_{\avoid\avoid}^{\{t\}} F_{sm}^{\{t\}}-F_{s\avoid}^{\{t\}}F_{\avoid m}^{\{t\}})(F_{\avoid\avoid}^{\{t\}}-F_{m\avoid}^{\{t\}})}{F_{\avoid\avoid}^{\{t\}}-F_{s\avoid}^{\{t\}}} \nonumber \\
	&=&\frac{F_{\avoid\avoid}^{\{t\}}F_{sm}^{\{t\}}-F_{s\avoid}^{\{t\}}F_{\avoid m}^{\{t\}}-F_{sm}^{\{t\}}F_{m\avoid}^{\{t\}}+Q_s^{\{\avoid,\overline{t}\}}F_{\avoid m}^{\{t\}}F_{m\avoid}^{\{t\}}}{F_{\avoid\avoid}^{\{t\}}-F_{s\avoid}^{\{t\}}}
	\end{eqnarray}
\end{relation}
\begin{proof}
	Use Eq. (\ref{eq:avoidanceF}) and Lemma (\ref{Nlemma1}).
\end{proof}

\begin{relation}
	Avoidance hitting time $H_s^{\{t,\overline{\avoid}\}}$ can be written in terms of classical fundamental matrix $F^{\{\avoid\}}$:
	\begin{equation}
	H_s^{\{t,\overline{\avoid}\}}=\sum_m F_{mt}^{\{\avoid\}}(\frac{F_{sm}^{\{\avoid\}}}{F_{st}^{\{\avoid\}}}-\frac{F_{tm}^{\{\avoid\}}}{F_{tt}^{\{\avoid\}}})
	\end{equation}
\end{relation}
\begin{proof}
	Use Relations (\ref{re:Fo}) and (\ref{re:HF1}).
\end{proof}

\begin{relation}
	Avoidance hitting time $H_s^{\{t,\overline{\avoid}\}}$ can be written in terms of classical fundamental matrix $F^{\{t\}}$:
	\begin{eqnarray}
	H_s^{\{t,\overline{\avoid}\}}&=&\frac{1}{F_{\avoid\avoid}^{\{t\}}}\frac{\sum_m ((F_{\avoid\avoid}^{\{t\}}F_{sm}^{\{t\}}-F_{s\avoid}^{\{t\}}F_{\avoid m}^{\{t\}})(F_{\avoid\avoid}^{\{t\}}-F_{m\avoid}^{\{t\}}))}{F_{\avoid\avoid}^{\{t\}}-F_{s\avoid}^{\{t\}}} \nonumber \\
	&=&\frac{1}{F_{\avoid\avoid}^{\{t\}}-F_{s\avoid}^{\{t\}}}(F_{\avoid\avoid}^{\{t\}}H_s^{\{t\}}-F_{s\avoid}^{\{t\}}H_{\avoid}^{\{t\}}-\sum_m F_{sm}^{\{t\}}F_{m\avoid}^{\{t\}}+Q_s^{\{\avoid,\overline{t}\}}\sum_m F_{\avoid m}^{\{t\}}F_{m\avoid}^{\{t\}}) \nonumber
	\end{eqnarray}
\end{relation}
\begin{proof}
	Use Relations (\ref{re:Ft}) and (\ref{re:HF1}).
\end{proof}

\section{Pivotality of Nodes in Reachability Problems} \label{sec:pivotal}

Reachability is crucial in any type of complex networks, be them
communication and computer networks, power grids, transportation
networks or social networks \cite{xie2005static}\cite{khakpour2010quantifying}\cite{parandehgheibi2013robustness}. More often than not, however, it is not
sufficient simply to know that a node $s$ can reach another
node $t$ in the network. Additional information is associated with
reachability such as how long (e.g., in terms of number of
intermediate nodes to be traversed or some other measures of time or
cost)  or how many possible ways (e.g., in terms of paths) for node $s$ to reach node $t$. Such information is essential for selecting paths for packet routing or information/commodity delivery, flow
scheduling, power management, traffic control, load balancing and
so forth in communication and computer networks, power grids and
transportation networks.
In this section, we analyze another piece of important information
associated with reachability -- which we call pivotality. Pivotality
captures how pivotal a role that a third node $k$ or a subset of nodes
$S$ may play in the reachability from node $s$ to node $t$ in a given
network by quantifying how many (and how long) paths from $s$ to
$t$ go through $k$ or $S$, and how many do not. We quantify this role by
exploiting relationships between the hitting time and
transit hitting times and examine how much of detour cost $k$ or $S$ can cause.
In particular,  we propose the \textit{avoidance-transit hitting time} pivotality metric (ATH). 
Finally we use several simulated and real-world
network examples to illustrate the
advantages and utility of avoidance and transit hitting times, especially in
comparison with existing metrics proposed in the literature.

\subsection{Related Work}
Closely related to what we study in this section, Ranjan and Zhang \cite{RanjanZhang13} introduce the notion of \textit{(forced) detour cost} of a random walker from a source $s$ to a target $t$ with respect to a third node $k$, which is defined as $\Delta H_{s}^t(k):=H_{s}^{\{k\}}+H_{k}^{\{t\}}-H_{s}^{\{t\}}$. Namely, the (forced) detour cost is the additional steps incurred when a random walker starts at source node $s$ and is forced to first visit the third  node $k$, and then starts from node $k$ to reach target node $t$  vs. the number of the steps it takes starting at source node $s$ and hitting target node $t$ for the first time. Ranjan and Zhang show~\cite{RanjanZhang13} that aggregated over all pairs of sources and targets, $\sum_{s}\sum_t \Delta H_{s}^t(k) =\mathcal{L}^+_{kk}$. Here $\mathcal{L}^+_{kk}$ is the diagonal entries of  $\mathcal{L}^+$,  the Penrose-Moore pseudo-inverse of the graph Laplacian $\mathcal{L} =D-A$, where $A=[a_{ij}]$ is the adjacency matrix of a graph (network) and $D=diag[d_i]$, $d_i=\sum_{j}a_{ij}$, is the diagonal degree matrix. Based on this (forced) detour cost as well as several other interpretations of the diagonal entries $\mathcal{L}^+_{kk}$ of $\mathcal{L}^+$, Ranjan and Zhang advocate  $C^*(k):=1/\mathcal{L}^{+}_{kk}$ as a new node centrality measure -- referred to as the {\em structural} or {\em topological} {\em centrality}, and demonstrate that $C^*(k):=1/\mathcal{L}^{+}_{kk}$ indeed better captures the structural/topological roles that node $k$ plays in a network than existing centrality metrics, in particular in terms of their roles in the overall network robustness. Motivated by the results in~\cite{RanjanZhang13}, in this paper we aim to provide a more precise characterization of how pivotal a role a third node $k$ may play in the random walks from a source node $s$ to a target node $t$ by probabilistically quantifying the number of paths from source $s$ to target $t$ that circumvent node $k$ vs. those that traverse node $k$ that the random walker is likely to take. This leads us to introduce two inter-related metrics, {\em avoidance} and {\em transit} hitting times, to measure the {\em pivotality} of node $k$ in the random walks from source $s$ to target $t$.

\subsection{Pivotality Metrics and Network Examples} 

For a given node $k$ with respect to a pair of source and target nodes $s$ and $t$, we define ATH as follows:
\begin{equation}
e_{ATH}(k)=H_{s}^{\{t\}}-H_{s}^{\{t,\breve{k}\}}= H_{s}^{\{t\}}-(H_{s}^{\{k,\overline{t}\}}+H_{k}^{\{t\}}). \label{eq:ATH}
\end{equation}
Note that if all paths from node $s$ to node $t$ go through a node $k^*$, then $e_{ATH}(k^*)=0$. In this case, $k^*$ is  the most ``pivotal'' point of any path from $s$ to $t$ in that all paths rely on $k^*$. We claim that in such a case, for any other node $k$, $e_{ATH}(k) \leq 0$; due to space limitation, we will omit the proof here.  In general,  $e_{ATH}(k)$ can be either positive, indicating that paths going through node $k$ are overall shorter than an ``average'' path from node $s$ to node $t$;  or negative, indicating that paths going through node $k$ are overall longer that an ``average'' path from node $s$ to node $t$.

For comparison, we also consider other metrics proposed in the literature.  We define the \textit{shortest-path} pivotality metric (SHP) to measure the pivotality of  node $k$ using the shortest paths only: $e_{SHP}(k)=L_{s}^{t}-(L_{s}^{k}+L_{k}^{t})$. The {\em maximum flow} pivotality metric (MF), $e_{MF}(k)$,  measures the amount of the maximum flow from $s$ to $t$ that goes through node $k$ in a flow network, where the weight of edges indicate their capacity. The (classical) {\em hitting time} pivotality
metric (CH) is defined as the negative of the (forced) detour cost defined in \cite{RanjanZhang13}, 
\begin{equation}
e_{CH}(k):=-\Delta H_s^{\{t\}}(k) =H_{s}^{\{t\}}-(H_{s}^{\{k\}}+H_{k}^{\{t\}}). \label{eq:CH}
\end{equation}
Notice the similarity between $e_{ATH}(k)$ and $e_{CH}(k)$, except the terms $H_{s}^{\{k,\overline{t}\}}$ and $H_{s}^{\{k\}}$. Due to the triangle inequality of the shortest path distance and the hitting time, $e_{SHP}(k) \leq 0$ and  $e_{CH}(k) \leq 0$ whereas by definition, $e_{MF}(k) \geq 0$ for all $k$ and all pairs of source and target nodes, $s$ and $t$.
Despite these differences, in terms of ranking of nodes based on their pivotality using each metric, what matters is their relative values: as long as $e(k_1)<e(k_2)$, node $k_2$ is more ``pivotal'' than $k_1$  in terms of  reachability from $s$ to $t$.

\begin{minipage}{\textwidth}
	\hspace{-10pt}
	\begin{minipage}[b]{0.35\textwidth}
		\centering
		\includegraphics[width=40mm]{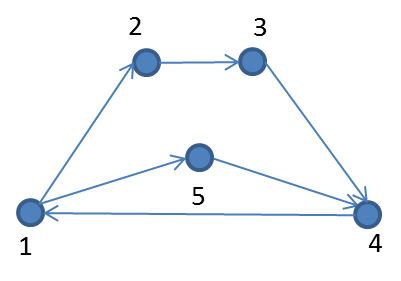}
		\captionof{figure}{Network example 1}\label{fig:toyExample1}
	\end{minipage}
	\hfill
	\begin{minipage}[b]{0.55\textwidth}
			\hspace{-10pt}
		\begin{tabular}{c|ccc}
	nodes & 2 & 3 & 5 \\
	\hline
	$e_{SHP}$ & -1 & -1 & 0 \\
	$e_{MF}$ & 0.5 & 0.5 & 0.5 \\
	$e_{CH}$ & -3.5 & -3.5 & -3.5 \\
	$e_{ATH}$ & -0.5 & -0.5 & 0.5 \\
\end{tabular}
		\captionof{table}{Pivotality metrics for \newline reachability from node 1 to node 4}\label{tb:toyexample1}
	\end{minipage}
\end{minipage}
\vspace*{10pt}

\begin{minipage}{\textwidth}
	\hspace{-15pt}
	\begin{minipage}[b]{0.35\textwidth}
		\centering
		\includegraphics[width=40mm]{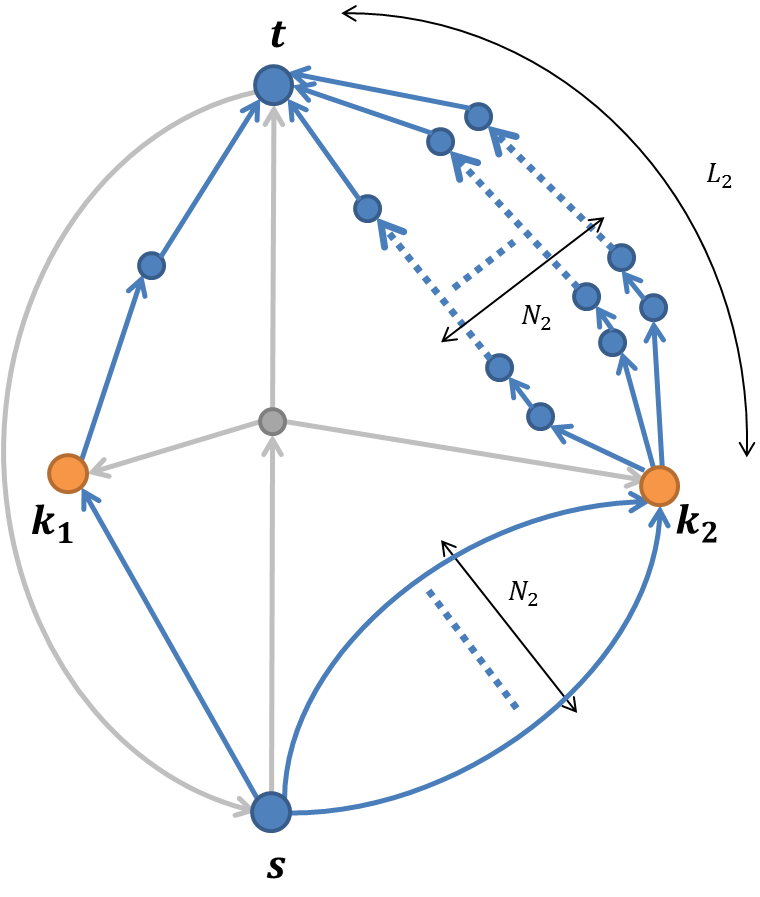}
		\captionof{figure}{Network example 2}\label{fig:toyExample2}
	\end{minipage}
	\hfill
	\begin{minipage}[b]{0.6\textwidth}
		\hspace{-15pt}
		\begin{tabular}{c|cc}
	& $e_{CH}$ & $e_{ATH}$\\
	& $k_1,k_2$ & $k_1,k_2$\\
	\hline
	$L_2=1,N_2=2$ & -7,-2.5 & -0.75,0.36 \\
	$L_2=2,N_2=1$ & -5.14,-5.14 & -0.14,-0.14 \\
	$L_2=2,N_2=2$ & -8.17,-2.92 & -0.17,-0.06 \\
	$L_2=20,N_2=2$ & -29.17,-10.42 & 10.33,-7.56 \\
	$L_2=20,N_2=1$ & -15.14,-15.14 & 7.86,-10.14 \\
	\end{tabular}
		\captionof{table}{CH and ATH pivotality metrics for\newline  various choices of $N_2$ and $L_2$}\label{tb:toyexample2}
	\end{minipage}
\end{minipage}
\vspace*{10pt}

Using several simple network examples, in this section we illustrate and compare the behavior of the pivotality metrics defined above. First consider the simple network example shown in Fig.~\ref{fig:toyExample1} where the weight of all edges is 1, i.e., $a_{ij}=1$. With node $1$ being the source and node $4$  the target, it is intuitively apparent that 
node 5 is more ``pivotal'' than node 2 or node 3, given that it is on the shorter path. The pivotaliy metrics computed using the four methods are shown in Table~\ref{tb:toyexample1}. We say that both the MF and CH metrics fail to rank the nodes correctly in that they are not able to recognize the higher pivotality of node 5 over nodes 2 and 3.

Figure~\ref{fig:toyExample2} provides a more general network example which can help illustrate the different behaviors of the pivotality metrics under study. 
In this network, there exists a shortest path of length 2 from source $s$ to target $t$ (gray-colored path) interconnected to two groups of (blue-colored) paths passing through $k_1$ and $k_2$: a three-hop path from source $s$ via node $k_1$ to target $t$, whereas there are $N_2$ parallel paths going through node $k_2$, the length of which are $L_2+1$. If $L_2=2$ and $N_2=1$ the network is symmetric with respect to $k_1$ and $k_2$ and yields equal pivotality for $k_1$ and $k_2$ in reachability from $s$ to $t$ (second row of Table~\ref{tb:toyexample2}). However, if $N_2 \approx 1$ and $L_2 \gg 2$, intuitively node $k_1$ plays a more pivotal role than $k_2$. On the other hand, as the number $N_2$ of parallel paths going through $k_2$ increases while their length  $L_2+1$ is not significantly much longer than 3, say, $L_2=3$, node $k_2$ will play an increasingly more pivotal role in delivering traffic, information or other commodity from node $s$ to node $t$. Intuitively, there is a trade-off between $N_2$ and $L_2$: more parallel paths going through node $k_2$ will increase its pivotality as it enhances the overall ``capacity'' from node $s$ to node $t$; however larger $L_2$ will diminish its pivotality as longer paths increase the ``cost'' of using these parallel paths.
Despite such intuitions regarding the relative pivotality values of node $k_1$ and  node $k_2$, if $L_2>2$ the SHP pivotality metric will always rank node $k_1$ higher than $k_2$ independently of $N_2$ (for $L_2=2$ gives the same ranking to them). Whereas,  as long as $N_2 >1$, the MF pivotality metric will always rank node $k_2$ higher than node $k_1$ independently of $L_2$. Hence both these two metrics fail to capture the differing roles of node $k_2$ with varying $N_2$ and $L_2$. To evaluate the performance of CH and ATH pivotality metrics in capturing the differing roles of node $k_2$ with varying $N_2$ and $L_2$, some example values are shown in Table~\ref{tb:toyexample2}. Based on these results, the CH pivotality metric ranks node $k_2$ higher than node $k_1$ as long as $N_2>1$, and ranks them the same when $N_2=1$ no matter how large is $L_2$, behaving the same as the MF pivotality metric. However, the ATH pivotality metric ranks successfully node $k_1$ higher than node $k_2$ when $N_2$ is close to 1 and $L_2$ is quite larger than 2.  

\begin{figure}
	\centering
	\includegraphics[width=0.6\textwidth]{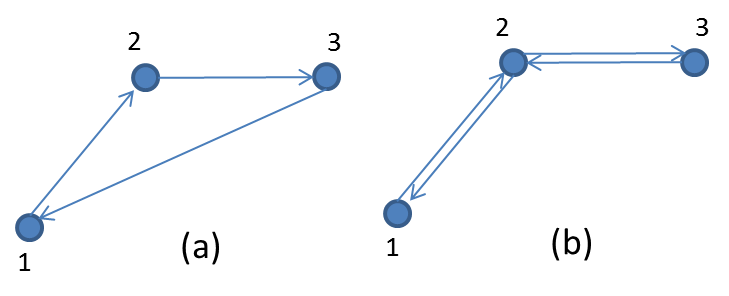}
	\caption{Network example 3}\label{fig:toyExample3}
\end{figure}

\begin{figure*}[htb]
	\centering
	\begin{minipage}[t]{0.45\textwidth}
		\begin{center}
			\includegraphics[width=0.9\textwidth]{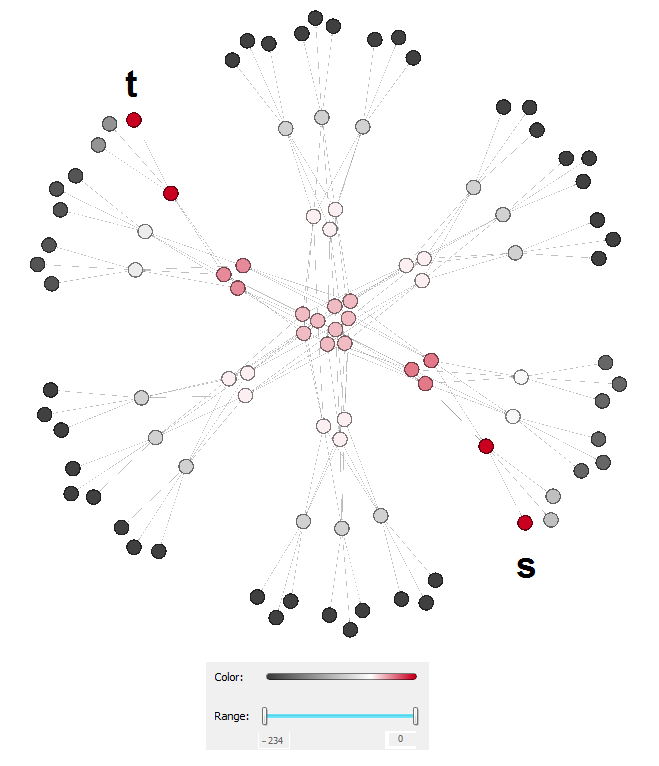}
			\caption{Node pivotality ranking in a Fat-tree network for the reachability of the source node $s$ to target node $t$:  red indicates highest pivotality and black shows non-pivotality.}\label{fig:fatTree}
		\end{center}
	\end{minipage}
	\hspace{2mm}
	\centering
	\begin{minipage}[t]{0.45\textwidth}
		\begin{center}
			\includegraphics[width=0.9\textwidth]{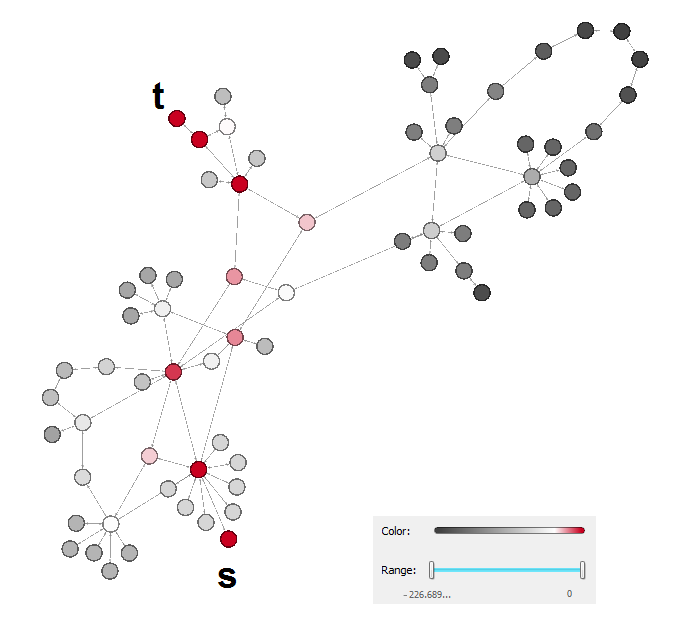}
			\caption{Node pivotality ranking in the  ESNet network for the reachability of the source node $s$ to target node $t$:  red indicates highest pivotality and black shows non-pivotality.}\label{fig:esnet}
		\end{center}
	\end{minipage}
	\hspace{2mm}
	\centering
\end{figure*}

The subtle difference in the behaviors of the CH and ATH pivotality metrics lies in the term $H_{s}^{\{k\}}$ in eq.(\ref{eq:CH}) vs. the term $H_{s}^{\{k,\overline{t}\}}$ in eq.(\ref{eq:ATH}). Namely, in accounting for the (forced) detour cost, the CH method allows  and includes paths/walks from the source node $s$ to the third node $k$ that may have already traversed the target node $t$; in network example 2, increasing $L_2$ has a destructive effect on the CH pivotality metric of $k_1$ by accounting the paths passing through $t$ before hitting $k_1$, such as the walk $(s-k_2-t-s-k_2-t-...-s-k_1)$, and increasing the term $H_{s}^{\{k_1\}}$ in eq.(\ref{eq:CH}) as the result. In contrast, the ATH method excludes such paths/walks in accounting for the detour cost. As a result, the ATH provides a more precise quantification of the detour cost when a random walker is ``forced'' to transit a third node $k$, and thereby how pivotal a role node $k$ plays in the reachability from a source to a target. 

The ATH metric allows us to identify nodes that are ``superfluous'' with respect to the reachability of a source to a target. This can be best illustrated by the two simple examples shown in Fig.~\ref{fig:toyExample3}. In both examples, consider node $1$ as the source and node $2$ as the target. It is obvious that node 3 is ``superfluous'' with respect to this source-target pair in that node $3$ plays no part in the reachability from node 1 to node 2. In other words, if node 3 fails or is removed from the network, the reachability from node 1 to node 2 (and the associated ``capacity'') is not affected at all. This can be captured by the fact that in both networks  in Figs.~\ref{fig:toyExample3} (a) and (b), 
the probability of hitting node 3 before node  2 is zero, i.e., $Q^{\{3,\overline{2}\}}_{1}=0$.
Thus the denominator of the term $H_{1}^{\{3,\overline{2}\}}$ in eq.(\ref{eq:avoidanceH}) becomes zero and thus  $H_{1}^{\{3,\overline{2}\}}=\infty$. This renders 
$e_{ATH}(3)=-\infty$ (see eq.(\ref{eq:ATH})), indicating the \emph{non-pivotality} of  node 3. In contrast, the CH metric and SHP metric yield $e_{CH}(3)=-3$ and $e_{SHP}(3)=-3$ for Fig.~\ref{fig:toyExample3}(a) and $e_{CH}(3)=-4$ and $e_{SHP}(3)=-2$ for Fig.~\ref{fig:toyExample3}(b) respectively.

\subsection{ Node Pivotality Ranking using the  ATH Metric}
Lastly we  apply the node pivotality ranking using our ATH metric to two real-world networks: Fat-Tree~\cite{leiserson1985fat} and the ESNet~\cite{esnet}. 
Fat-tree is a special $h$-ary ($h \geq 2$) ``tree-shaped'' structure first proposed in~\cite{leiserson1985fat} for efficient communication with uniform bi-section bandwidth, and for this reason  it has been adopted in data center networks~\cite{fat-tree-dc}.  Fig. (\ref{fig:fatTree}) shows 3-ary fat-tree structure with 99 nodes, where the node colors are shaded based on their ATH pivotality measures with respect to the reachability from the source $s$ to the target node $t$. In the figure,  the color spectrum from red to white and then to black shows the range of the ATH value from high to low: the nodes with the  larger ATH value, are more pivotal to the reachability from $s$ to $t$ are  represented with red and ``reddish'' colors; in contrast,   the nodes that play no part in the reachability from $s$ to $t$ are represented with black color. The results for the ESNet, the DoE energy science network with 68 nodes~\cite{esnet} are shown in  Fig.~(\ref{fig:esnet}). Both examples illustrate the efficacy of the ATH metric in correctly capturing and ranking the pivotality of nodes in the reachability from a source node to a target node. Due to space limitation, we do not elaborate on them.

\section*{Acknowledgement}
The research was supported in part by  US DoD DTRA grants HDTRA1-09-1-0050
and HDTRA1-14-1-0040, and ARO MURI Award W911NF-12-1-0385.

\newpage
\appendix
\section{More on Avoidance and Transit Random Walk Metrics}
The following two lemmas would be used for the proof of theorems later in the appendix. For the proofs of these lemmas please refer to \cite{golnari2017random}.
\begin{lemma}[Incremental Computation of Fundamental Matrix] 
	\label{Nlemma1}
	The fundamental matrix for target set of $\mathcal{O}_1\cup \mathcal{O}_2$ can be computed from the fundamental matrix for target set $\mathcal{O}_1$:
	\begin{equation} \label{eq:Nlemma1}
	F_{im}^{\{\mathcal{O}_1, \mathcal{O}_2\}}=F_{im}^{\{\mathcal{O}_1\}}-F_{i\mathcal{O}_2}^{\{\mathcal{O}_1\}}({F_{\mathcal{O}_2\mathcal{O}_2}^{\{\mathcal{O}_1\}}})^{-1}F_{\mathcal{O}_2m}^{\{\mathcal{O}_1\}},
	\end{equation}
	where the subscripts represent the rows and columns selected from the matrix respectively, e.g. $F_{i\mathcal{O}_2}^{\{\mathcal{O}_1\}}$ denotes the row $i$ and the columns corresponding to set $\mathcal{O}_2$ of the fundamental matrix $F^{\{\mathcal{O}_1\}}$.
\end{lemma}

\begin{lemma}[Absorption Probability and Normalized Fundamental Matrix] \label{Nlemma3}
	The absorption probability for absorbing set $\{j\}\cup \mathcal{O}$ can be found from the fundamental matrix for absorbing set $\mathcal{O}$:
	\begin{equation} 
	Q_i^{\{j,\overline{\mathcal{O}}\}}=\frac{F_{ij}^{\{\mathcal{O}\}}}{F_{jj}^{\{\mathcal{O}\}}}
	\end{equation}
\end{lemma}

\subsection{Avoidance Fundamental Matrix} \label{apx:avoidFundMat}
The indicator function $1_{\{X_k=m\}}$ is the random variable equal
to 1 if $X_k=m$ and 0 otherwise. The number of visits $\nu_m$ to
$m$ is written in terms of indicator functions as $\nu_m=\sum_{k=0}^{\infty} 1_{\{X_k=m\}}$.
The avoidance fundamental matrix (or avoidance number of visits), which we introduce as the conditional expectation of number of visits conditioned on avoiding set $\mathcal{O}$, is obtained as follows. Recall that $\kappa_t$ is the stopping criteria for the walk. 
\begin{eqnarray}
F_{sm}^{\{t,\overline{\mathcal{O}}\}}&=&\mathbb{E}_s(\nu_m|X_{i\leq \kappa_t}\notin \mathcal{O})= \sum_{k=0} \mathbb{E}_s(1_{\{X_k=m\}}|X_{i\leq \kappa_t}\notin \mathcal{O}) \nonumber
\\&=& \sum_{k=0} \mathbb{P}(X_k=m|X_{i\leq \kappa_t}\notin \mathcal{O},X_0=s) \nonumber
\\&=& \frac{\sum_{k=0}\mathbb{P}(X_k=m,X_{i\leq \kappa_t}\notin \mathcal{O}|X_0=s)}{\mathbb{P}(X_{i\leq \kappa_t}\notin \mathcal{O}|X_0=s)} \nonumber
\\&=& \frac{\sum_{k=0}\mathbb{P}(X_k=m,X_{i< k}\notin \mathcal{O},X_{k<i\leq \kappa_t}\notin \mathcal{O}|X_0=s)}{\mathbb{P}(X_{i\leq \kappa_t}\notin \mathcal{O}|X_0=s)} \nonumber
\\&=& \frac{\sum_{k=0}\mathbb{P}(X_k=m,X_{i< k}\notin \mathcal{O}|X_0=s)\mathbb{P}(X_k=m,X_{k<i\leq \kappa_t}\notin \mathcal{O}|X_0=s)}{\mathbb{P}(X_{i\leq \kappa_t}\notin \mathcal{O}|X_0=s)} \nonumber
\\&=& \frac{\sum_{k=0}\mathbb{P}(X_k=m,X_{i< k}\notin \mathcal{O}|X_0=s)\mathbb{P}(X_{0<i\leq \kappa_t}\notin \mathcal{O}|X_0=m)}{\mathbb{P}(X_{i\leq \kappa_t}\notin \mathcal{O}|X_0=s)} \nonumber
\\&=& \frac{\sum_{k=0}[P_{\mathcal{TT}}^{k}]_{sm}\sum_{k=1} [P_{\mathcal{TT}}^{k-1}P_{\mathcal{TA}}]_{mt}}{\sum_{k=1} [P_{\mathcal{TT}}^{k-1}P_{\mathcal{TA}}]_{st}} \label{eq:general_avoidanceF} \nonumber
\\&=& \frac{F_{sm}^{\{t,\mathcal{O}\}}\sum_{k=1} [P_{\mathcal{TT}}^{k-1}P_{\mathcal{TA}}]_{mt}}{\sum_{k=1} [P_{\mathcal{TT}}^{k-1}P_{\mathcal{TA}}]_{st}}, 
\end{eqnarray}

Now if terms $\sum_{k=1} [P_{\mathcal{TT}}^{k-1}P_{\mathcal{TA}}]_{it}$, for $i=m,s$ in numerator and denominator, are replaced with the following relation, the simplified expression in (\ref{eq:avoidanceF}) is obtained:
\begin{eqnarray}
\sum_k [P^{k-1}_{\mathcal{TT}}P_{\mathcal{TA}}]_{it} 
&=& u'_i( I+ P_{\mathcal{TT}}+P^{2}_{\mathcal{TT}}+...) P_{\mathcal{TA}}u_t \nonumber
\\&=&u'_i(I-P_{\mathcal{TT}})^{-1}P_{\mathcal{TA}}u_t \nonumber
\\&=&u'_i F P_{\mathcal{TA}}u_t \nonumber
\\&=&u'_i Q^{\{t,\overline{\mathcal{O}}\}} \nonumber
\\&=& Q_i^{\{t,\overline{\mathcal{O}}\}}, \label{eq:PP}
\end{eqnarray}
where $u_i$ is a column vector of all zeros but $i$-th entry equal to 1.

\subsection{Avoidance Hitting Time}\label{apx:avoidHitTime}
The hitting time of a node $t\in V$ is the random variable $\kappa_{t}:\Omega\rightarrow\{0,1,2,...\}\cup\{\infty\}$
given by
$\kappa_{t}=\inf{\{\kappa\geq 0: X_{\kappa}=t\}}$,
where we agree that the infimum of the empty set $\emptyset$ is $\infty$. The hitting time $\kappa_{t}$ represents the number of steps that the walk takes until it hits $t$ for the first time. The avoidance (expected) hitting time from $s$ to $t$ conditioned on avoiding set $\mathcal{O}$ is defined as follows
\begin{eqnarray}
H_s^{\{t,\overline{\mathcal{O}}\}}=\mathbb{E}_s[\kappa_t|X_{i\leq \kappa_t}\notin \mathcal{O}]&=&\sum_{k=1} k\mathbb{P}(X_k=t|X_{i\leq \kappa_t}\notin \mathcal{O},X_0=s) \nonumber
\\&=&\sum_{k=1} k \frac{\mathbb{P}(X_k=t,X_{i\leq k}\notin \mathcal{O}|X_0=s)}{\mathbb{P}(X_{i\leq \kappa_t}\notin \mathcal{O}|X_0=s)} \nonumber
\\&=&\frac{\sum_{k=1} k \mathbb{P}(X_k=t,X_{i\leq k}\notin \mathcal{O}|X_0=s)}{\mathbb{P}(X_{i\leq \kappa_t}\notin \mathcal{O}|X_0=s)} \nonumber
\\&=&\frac{\sum_{k=1} k \mathbb{P}(X_k=t,X_{i\leq k}\notin \mathcal{O}|X_0=s)}{\sum_{k=1}\mathbb{P}(\kappa_t=k,X_{i\leq k}\notin \mathcal{O}|X_0=s)} \nonumber
\\&=&\frac{\sum_{k=1} k \mathbb{P}(X_k=t,X_{i\leq k}\notin \mathcal{O}|X_0=s)}{\sum_{k=1}\mathbb{P}(X_k=t,X_{i\leq k}\notin \mathcal{O}|X_0=s)} \nonumber 
\\&=&\frac{\sum_{k=1} k[P_{\mathcal{TT}}^{k-1}P_{\mathcal{TA}}]_{st}}{\sum_{k=1} [P_{\mathcal{TT}}^{k-1}P_{\mathcal{TA}}]_{st}}, \label{eq:general_avoidanceH}
\end{eqnarray}
here the transient set is $\mathcal{T}=V\setminus (\mathcal{O}\cup\{t\})$. Substituting the numerator with the following relation (\ref{eq:kPP}) and the denominator with (\ref{eq:PP}), the simplified expression for avoidance hitting time in (\ref{eq:avoidanceH}) is obtained.
\begin{eqnarray}
\sum_k k[P^{k-1}_{\mathcal{TT}}P_{\mathcal{TA}}]_{st} \nonumber
&=& u'_s( I+ 2P_{\mathcal{TT}}+3P^{2}_{\mathcal{TT}}+...) P_{\mathcal{TA}}u_t \nonumber
\\&=&u'_s(I-P_{\mathcal{TT}})^{-2}P_{\mathcal{TA}}u_t \nonumber
\\&=&u'_s F^2 P_{\mathcal{TA}}u_t \nonumber
\\&=&u'_s F Q^{\{t,\overline{\mathcal{O}}\}} \nonumber
\\&=& \sum_m F_{sm}^{\{t,\mathcal{O}\}} Q_m^{\{t,\overline{\mathcal{O}}\}} \label{eq:kPP}
\end{eqnarray}

The relation between avoidance hitting time and avoidance fundamental matrix is as follows:
\begin{equation} \label{h=N1}
H^{\{t,\overline{\mathcal{O}}\}}_s=\sum_{m}F^{\{t,\overline{\mathcal{O}}\}}_{sm}=F^{\{t,\overline{\mathcal{O}}\}}_{s:}\textbf{1},
\end{equation}
where $F^{\{t,\overline{\mathcal{O}}\}}_{s:}$ is the $s$-th row of avoidance fundamental matrix.

\subsection{Avoidance Hitting Cost} \label{apx:avoidHitCost}
Let $G={(X_k)}_{k>0}$ be a discrete-time Markov chain with transition matrix $P$ and weight matrix $W$. The hitting cost of a node $t\in V$ is a random variable $\eta_{t}:\Omega\rightarrow \mathcal{C}$ given by
$\eta_{t}=\inf{\{\eta\geq 0: \exists k, X_k=t, \sum_{i=1}^k w_{X_{i-1}X_i}=\eta\}}$, where $\mathcal{C}$ is a countable set. Let $Z_{st}^{\overline{\mathcal{O}}}$ be the set of all walks from $s$ to $t$ which avoid set $\mathcal{O}$, and $\zeta_j$ be the $j$-th walk from this set. We also use $Z_{sm}^{\overline{\mathcal{O}}}(l)$ to denote the subset of walks whose total length is $l$, and $Z_{sm}^{\overline{\mathcal{O}}}(k,l)$ to specify the walks which have total length of $l$ and total step size of $k$. Avoidance (expected) hitting cost from $s$ to $t$ conditioned on avoiding $\mathcal{O}$ is defined as follows:

\begin{eqnarray}
U_s^{\{t,\overline{\mathcal{O}}\}}&=&\mathbb{E}_s[\eta_t|X_k=t, X_{i\leq k}\notin \mathcal{O}]=\sum_{l\in\mathcal{C}} l\mathbb{P}(\eta_t=l|X_k=t, X_{i\leq k}\notin \mathcal{O}, X_0=s) \nonumber
\\&=&\frac{\sum_{l\in\mathcal{C}} l \mathbb{P}(\eta_t=l,X_k=t, X_{i\leq k}\notin \mathcal{O}| X_0=s)}{\mathbb{P}(X_k=t, X_{i\leq k}\notin \mathcal{O}| X_0=s)}
\\&=&\frac{\sum_{l\in\mathcal{C}} l\sum_{k=1}^{<\infty}\mathbb{P}(\sum_{i=1}^k w_{X_{i-1}X_i}=l,X_k=t, X_{i\leq k}\notin \mathcal{O}|X_0=s)}{\sum_{k=1}^{<\infty}\mathbb{P}(X_k=t,X_{i\leq k}\notin \mathcal{O}|X_0=s)} \nonumber
\\&=& \frac{\sum_{l\in\mathcal{C}}l\sum_{k=1}^{<\infty}\sum_{\zeta_j\in Z_{st}^{\overline{\mathcal{O}}}(k,l)} \textrm{Pr}_{\zeta_j}}{\sum_{k=1}^{<\infty}\sum_{\zeta_j\in Z_{st}^{\overline{\mathcal{O}}}(k)} \textrm{Pr}_{\zeta_j}} \label{eq:avoidUform2}
\\&=& \frac{\sum_{l\in\mathcal{C}}l\sum_{\zeta_j\in Z_{st}^{\overline{\mathcal{O}}}(l)} \textrm{Pr}_{\zeta_j}}{\sum_{\zeta_j\in Z_{st}^{\overline{\mathcal{O}}}} \textrm{Pr}_{\zeta_j}} \nonumber
\\&=& \frac{\sum_{l\in\mathcal{C}}l\sum_{\zeta_j\in Z_{st}^{\overline{\mathcal{O}}}(l)} \textrm{Pr}_{\zeta_j}}{\sum_{l\in\mathcal{C}}\sum_{\zeta_j\in Z_{st}^{\overline{\mathcal{O}}}(l)} \textrm{Pr}_{\zeta_j}} \nonumber
\\&=& \frac{\sum_{l\in\mathcal{C}}l \textrm{Pr}_l^{\overline{\mathcal{O}}}}{\sum_{l\in\mathcal{C}} \textrm{Pr}_l^{\overline{\mathcal{O}}}} \label{eq:general_avoidanceU}
\end{eqnarray}
where $\textrm{Pr}_l^{\overline{\mathcal{O}}}$ is the probability of hitting $t$ in total length of $l$ when starting from $s$ and avoiding set $\mathcal{O}$. It is obtained from the aggregation of walk probabilities with length $l$ which avoid set $\mathcal{O}$. Therefor, the following three quantities are all the same: $\textrm{Pr}_l^{\overline{\mathcal{O}}}=\sum_{\zeta_j\in Z_{st}^{\overline{\mathcal{O}}}(l)} \textrm{Pr}_{\zeta_j}=\mathbb{P}(\eta_t=l|X_{i\leq k}\notin \mathcal{O}, X_k=t, X_0=s)$. 

We can also continue (\ref{eq:avoidUform2}) as an aggregation over walks with specified lengths to achieve another form of avoidance hitting cost and derive the form presented in (\ref{eq:avoidanceU}):
\begin{eqnarray}
U_s^{\{t,\overline{\mathcal{O}}\}}&=&\frac{\sum_{l\in\mathcal{C}}l\sum_{k=1}^{<\infty}\sum_{\zeta_j\in Z_{st}^{\overline{\mathcal{O}}}(k,l)} \textrm{Pr}_{\zeta_j}}{\sum_{k=1}^{<\infty}\sum_{\zeta_j\in Z_{st}^{\overline{\mathcal{O}}}(k)} \textrm{Pr}_{\zeta_j}}
\\&=& \frac{\sum_{l\in\mathcal{C}}\sum_{k=1}^{<\infty}\sum_{\zeta_j\in Z_{st}^{\overline{\mathcal{O}}}(k,l)}l_{\zeta_j} \textrm{Pr}_{\zeta_j}}{\sum_{k=1}^{<\infty}\sum_{\zeta_j\in Z_{st}^{\overline{\mathcal{O}}}(k)} \textrm{Pr}_{\zeta_j}} \nonumber
\\&=& \frac{\sum_{\zeta_j\in Z_{st}^{\overline{\mathcal{O}}}}l_{\zeta_j} \textrm{Pr}_{\zeta_j}}{\sum_{\zeta_j\in Z_{st}^{\overline{\mathcal{O}}}} \textrm{Pr}_{\zeta_j}}
\\&=&\frac{\sum_{\zeta_j\in Z_{st}^{\overline{\mathcal{O}}}} \textrm{Pr}_{\zeta_j} \sum_{k=1}^{k_{\zeta_j}} w_{v_{k-1}v_k}}{Q_s^{\{t,\overline{\mathcal{O}}\}}} \nonumber
\\&=& \frac{\sum_{\zeta_j\in Z_{st}^{\overline{\mathcal{O}}}}\sum_{k=1}^{k_{\zeta_j}}[\prod_{i=1}^{k}P_{v_{i-1}v_i}(P_{v_kv_{k+1}}w_{v_kv_{k+1}})\prod_{i=k+2}^{k_{\zeta_j}}P_{v_{i-1}v_i}]}{Q_s^{\{t,\overline{\mathcal{O}}\}}} \nonumber
\\&=& \frac{\sum_{e_{xy}\in E,x\in \mathcal{T},y\in\mathcal{T}\cup\{t\}}P_{xy}w_{xy}(\sum_{\zeta_j\in Z_{sx}^{\overline{\mathcal{O}}}}\textrm{Pr}_{\zeta_j})\cdot(\sum_{\zeta_i\in Z_{yt}^{\overline{\mathcal{O}}}}\textrm{Pr}_{\zeta_i})}{Q_s^{\{t,\overline{\mathcal{O}}\}}} \nonumber
\\&=& \frac{\sum_{e_{xy}\in E,x\in \mathcal{T},y\in\mathcal{T}\cup\{t\}}P_{xy}w_{xy}(\sum_k\sum_{\zeta_j\in Z_{sx}^{\overline{\mathcal{O}}}(k)}\textrm{Pr}_{\zeta_j})\cdot(\sum_k\sum_{\zeta_i\in Z_{yt}^{\overline{\mathcal{O}}}(k)}\textrm{Pr}_{\zeta_i})}{Q_s^{\{t,\overline{\mathcal{O}}\}}} \nonumber
\\&=& \frac{\sum_{e_{xy}\in E,x\in \mathcal{T},y\in\mathcal{T}\cup\{t\}}P_{xy}w_{xy}(\sum_{\substack{k}} [P_{\mathcal{TT}}^k]_{sx})\cdot(\sum_{\substack{k}} [P_{\mathcal{TT}}^{k-1}P_{\mathcal{TA}}]_{yt})}{Q_s^{\{t,\overline{\mathcal{O}}\}}} \nonumber
\\&=&\frac{\sum_{e_{xy}\in E,x\in \mathcal{T},y\in\mathcal{T}\cup\{t\}}P_{xy}w_{xy}F_{sx}^{\{t,\mathcal{O}\}}Q_y^{\{t,\overline{\mathcal{O}}\}}}{Q_s^{\{t,\overline{\mathcal{O}}\}}} \nonumber
\\&=&\frac{\sum_{x\in \mathcal{T}} F_{sx}^{\{t,\mathcal{O}\}} Q_x^{\{t,\overline{\mathcal{O}}\}} \sum_{y\in \mathcal{N}_{out}(x)\setminus \mathcal{O}} p_{xy}\frac{Q_y^{\{t,\overline{\mathcal{O}}\}}}{Q_x^{\{t,\overline{\mathcal{O}}\}}}w_{xy}}{Q_s^{\{t,\overline{\mathcal{O}}\}}} \nonumber
\\&=&\frac{\sum_{x\in \mathcal{T}} F_{sx}^{\{t,\mathcal{O}\}} Q_x^{\{t,\overline{\mathcal{O}}\}} r^{t,\overline{\mathcal{O}}}_x}{Q_s^{\{t,\overline{\mathcal{O}}\}}},
\end{eqnarray}		
where the transient set here is  equal to $\mathcal{T}=V\setminus\mathcal{O}\cup\{t\}$.

\newpage
  \bibliographystyle{elsarticle-num} 
  \bibliography{main}


%
%
%
\end{document}